\documentclass[letterpaper]{article}
\usepackage{amsmath}
\usepackage{amsthm}
\usepackage{amssymb}
\usepackage{graphicx}
\usepackage{hyperref}
\usepackage{natbib}
\usepackage{graphicx}
\usepackage{algorithm}
\usepackage[noend]{algorithmic}
\usepackage{lipsum}
\usepackage{array}
\usepackage{booktabs}
\usepackage{caption}
\usepackage{subcaption}
\usepackage{mdframed}
\usepackage{xspace}
\usepackage{siunitx}

\newcounter{shared}

\newtheorem{theorem}[shared]{Theorem}
\newtheorem{lemma}[shared]{Lemma}
\newtheorem{observation}[shared]{Observation}
\newtheorem{proposition}[shared]{Proposition}

\makeatletter
\newcounter{procedurecounter}
\newenvironment{procedure}[1][htb]
  {\renewcommand{\ALG@name}{Procedure}
   \floatname{algorithm}{Procedure}
   
   \setcounter{algorithm}{\value{procedurecounter}}
   \begin{algorithm}[#1]}{\end{algorithm}
   \stepcounter{procedurecounter}}
\makeatother

\newcommand{\eps}{\varepsilon}
\newcommand{\newKLM}{L-KLM\@\xspace}
\newcommand{\E}{\mathbb{E}}
\newcommand{\given}{\, | \,}

\title{Scalable Algorithms for Approximate \\DNF Model Counting}
\author{Paul Burkhardt\footnote{National Security Agency, \url{pburkha@nsa.gov}}~~and
    David G. Harris\footnote{University of Maryland, \url{davidgharris29@gmail.com}}~~and
    Kevin T. Schmitt\footnote{National Security Agency, \url{ktschm2@nsa.gov}}}
\date{}

\begin{document}
\maketitle

% -------------------------- abstract ---------------------------------
\begin{abstract}

Model counting of Disjunctive Normal Form (DNF) formulas is a critical problem in applications such as probabilistic inference and network reliability. 
For example, it is often used for query evaluation in probabilistic databases. Due to the computational intractability of exact DNF counting, there has been a 
line of research into a variety of approximation algorithms.

We develop a new Monte Carlo approach with an adaptive stopping rule and short-circuit formula evaluation. We prove it achieves Probably Approximately Correct (PAC) bounds and has 
asymptotically improved time and randomness complexity compared to previous methods. We also show experimentally that it outperforms prior algorithms by at least three orders of magnitude in running time and scalability.

\end{abstract}

% ------------------------ body of paper ------------------------------
\section{Introduction}

The \emph{Model Counting Problem}, that is, determining the number of satisfying assignments to a constraint 
satisfaction problem, is a crucial
challenge in a range of disciplines. It has numerous practical applications
including probabilistic inference (\cite{inference,sang,bacchus}), 
probabilistic databases (\cite{probdb,pdont}), probabilistic programming (\cite{problog,aproblog}), 
AI explainability (\cite{AIexplain}), hardware verification (\cite{CSverify}),
and network reliability analysis 
(\cite{kargernet,duenasnet,kargercut}). 

We consider the setting of model counting for formulas 
given in Disjunctive Normal Form (DNF). This problem is \#P-complete to compute exactly (\cite{CountCplx,ADDmc}), but classical results of \cite{KL} and \cite{KLM} provide an efficient Fully-Polynomial Randomized Approximation Scheme 
(FPRAS). We refer to the latter algorithm as \emph{KLM}. Formally, these algorithms take as input a given DNF formula $\Phi$, as well as desired precision and confidence parameters 
$(\eps, \delta)$, and run in polynomial-time to produce an estimate $\hat \mu$ such that

\[\Pr(\mu (1 - \eps) \leq \hat{\mu} \leq \mu(1 + \eps)) \geq 1 - \delta\]

\noindent where $\mu$ is the (weighted) proportion of assignments which satisfy $\Phi$.\footnote{In some references, a default value of $\delta = 1/4$ is used; one can always obtain a desired smaller value of $\delta$ via median-amplification.} 
This is a key difference between DNF formulas and more general formula classes such as Conjunctive Normal Form (CNF), 
for which estimation is also NP-hard (\cite{sly,approxmc2025}).

A number of FPRAS algorithms have been developed for DNF model counting, e.g. \cite{DAppMC,SDAppMC}.
The most recent is an algorithm called Pepin (\cite{pepin}), largely inspired by streaming methods of \cite{CVM}.
It also uses ideas from hashing-based CNF counting (\cite{approxmc}) and lazy sampling (\cite{lsample}).
Though the Pepin algorithm is primarily designed for streaming problems, it also claims better performance 
in the static setting compared to previous FPRAS algorithms. 

We also mention a heuristic approach called Neural\#DNF proposed by \cite{GNN}, based on training neural networks on top of the FPRAS 
algorithms. Neural\#DNF achieves a speed-up of three orders of magnitude compared to state-of-the-art FPRAS. Notably, it can
extend the range of feasible problem sizes to 15,000 variables, a $1.5\times$ improvement in scalability. Unfortunately, this approach does not have any guarantee of correctness. 

Our major contribution is to introduce a new Monte Carlo algorithm for approximate DNF counting. This algorithm features an adaptive 
stopping rule and short-circuit evaluation of the formula, which play a similar role to the “self-adjusting” feature of the KLM algorithm. The key difference is that our short-circuit evaluation allows a fixed ordering of clauses throughout all trials; as a result, our Main Algorithm has better memory access patterns and less random-number generation, leading to at least three orders of magnitude improvement in running time and scalability. This is a remarkable advance compared to Neural\#DNF's $1.5\times$ scaling.

In Table 1, we give running times for $\eps = 0.05$ and scaled from $10^3$ to $10^8$ variables and clauses.
We include the 
performance of our Main Algorithm, as well as the original KLM algorithm, the Pepin algorithm, and a novel variant of KLM (which we call \newKLM) we developed using lazy sampling. Our implementations of KLM, \newKLM, and our Main Algorithm used the same data structures where possible and were implemented without low-level optimizations. We later include a comparison with Neural\#DNF, using their reported times with the same parameters and benchmark generation.

\begin{table}[ht]
\centering
\caption{Running time in seconds by number of variables $n$ with $\eps=\delta=0.05$ and $m=n$. 
Timeout was $1$ day (86,400 seconds).}
\begin{tabular}{c|S[table-format=3.1,group-minimum-digits=4,group-separator={,}]
    S[table-format=4.1,group-minimum-digits=4,group-separator={,}]
    S[table-format=5.1,group-minimum-digits=4,group-separator={,}]
    S[table-format=4.1,group-minimum-digits=4,group-separator={,}]
    S[table-format=5.1,group-minimum-digits=4,group-separator={,}]
    S[table-format=5.1,group-minimum-digits=4,group-separator={,}]}
\toprule
\multicolumn{1}{c}{} & \multicolumn{6}{c}{$n$} \\
\cmidrule(lr){2-7}
\multicolumn{1}{c}{} & {$10^3$} & {$10^4$} & {$10^5$} & {$10^6$} & {$10^7$} & {$10^8$} \\
\midrule
Pepin & 288.0 & 5578.5 & {$ $} & {$ $} & {$ $} & {$ $} \\
KLM & 4.0 & 244.9 & 19836.0 & {$ $} & {$ $} & {$ $} \\
\newKLM & 0.8 & 7.9 & 111.9 & 2175.1 & 27103.7 & {$ $} \\
Main & 0.1 & 1.2 & 12.1 & 138.5 & 1582.8 & 22901.6 \\
\bottomrule
\end{tabular}
\end{table}

In Table 2, we summarize the asymptotic bounds of the algorithms. 
In addition to the time complexity, we include the randomness complexity (i.e., the number of random bits drawn) and the space complexity. 
See Theorems~\ref{main-analysis-theorem} and \ref{main-analysis2} for full details.  These bounds are given in terms of 
the number of clauses $m$, the number of variables $n$, the average clause width $w$, and a parameter $p \in [1/m,1]$ measuring the overlap between 
clauses which we will discuss in greater detail later. 
\begin{table}[ht]
\small
\centering
\caption{DNF counter complexities, up to constant factors. Some second-order terms have been omitted, for readability. 
 }
\renewcommand{\arraystretch}{1.75}
\begin{tabular}{c|p{3cm}p{4cm}p{1.5cm}}
\hline
Algorithm & Time & Randomness & Space \\
\hline
KLM & $\frac{\log1/\delta}{\eps^2}\cdot (mw + n/p)$ & $\frac{\log1/\delta}{\eps^2}\cdot \left(m \log m + n/p \right)$ & $mw$  \\
Pepin & $\frac{\log1/\delta}{\eps^2}\cdot mn\log\frac{m}{\delta\eps}$ & unspecified 
     & $mn \cdot \frac{\log1/\delta}{\eps^2}$ \\        
\newKLM & $\frac{\log1/\delta}{\eps^2}\cdot mw$ & $\frac{\log1/\delta}{\eps^2}\cdot m \log m$ & $mw$ \\
Main & $\frac{\log1/\delta}{\eps^2}\cdot mw\log\frac{1}{p} $ &  $\frac{\log1/\delta}{\eps^2}\cdot \min\{m \log \frac{1}{p},\frac{n}{p}\}$ & $mw$  \\
\hline
\end{tabular}
\end{table}

As shown empirically, our approaches outperform
existing state-of-the-art methods by a factor of at least three orders of magnitude in both running time and problem size scaling.
From a theoretical point of view, our new algorithm has asymptotic advantages over prior algorithms as well as L-KLM. It
appears to have comparable running time to Neural\#DNF, even without taking into account the latter's pretraining time. 
Furthermore, it has theoretical worst-case guarantees of the FPRAS algorithms. 

\section{Preliminaries and notation}
Let $V$ be a set of $n$ propositional variables. A \emph{literal} is an expression of the form $v$ or $\neg v$ for $v \in V$. A formula is in \emph{disjunctive normal form} (DNF) if it is a disjunction of clauses which are conjunctions of literals. In such a setting, we view a clause $C$ as a set of literals and a formula $\Phi$ as a set of $m$ clauses $C_1, \dots, C_m$. We consider only clauses which are non-empty and free of contradictions.

The \emph{width} of a clause $C$, denoted $W(C)$, is its cardinality as a set, i.e. the number of literals it contains. For a set of clauses $\mathcal C$, we write $W(\mathcal C) = \sum_{C \in \mathcal C} W(C)$. We denote the \emph{average clause width} by $w = \tfrac{1}{m} \sum_{C \in \Phi} W(C) = \frac{W(\Phi)}{m}$.  Note that storing $\Phi$ would require $m w$ space. We assume that $n \leq m w$, as otherwise there would be some variable not involved in any clause.

An \emph{assignment} is a set which contains exactly one of the literals $v, \neg v$ for each variable $v$. There are $2^n$ total possible assignments. An assignment $\nu$ \emph{satisfies} a clause $C$ if $C \cap \nu = C$. The assignment $\nu$ \emph{models} $\Phi$, denoted by $\nu\models\Phi$, if $\nu$ satisfies some clause $C\in\Phi$. The \emph{model count} of $\Phi$ is the total number of assignments which satisfy $\Phi$.

We also adopt a probabilistic interpretation of weighted model count: suppose we have a weight function $\rho: V \rightarrow [0,1]$, and we draw $\nu$ by including each positive literal $v$ with probability $\rho(v)$, otherwise we include the negative literal $\neg v$ with probability $\rho(\neg v) = 1-\rho(v)$. Each variable $v$ is handled independently; thus, the probability of an assignment $\nu$ is given by $\rho(\nu) =  \prod_{a \in \nu} \rho(a)$. 

We define the \emph{weighted model ratio} to be $\mu = \sum_{\nu \models \Phi} \rho(\nu)$, that is, the probability that $\Phi$ is true under the random variable assignment $\rho$. The unweighted model count is simply $2^n$ times the weighted model ratio for the special case where $\rho(v) = 1/2$ for all variables $v$.

We extend our notation by writing $\rho(C) = \prod_{a \in C}\rho(a)$ for a clause $C$. For a set of clauses $\mathcal C$, we also write $\rho(\mathcal C) = \sum_{C \in \mathcal C} \rho(C)$.

All instances of $\log$ are taken in the natural base $e$.

\section{Lazy Monte Carlo Sampling}
If we simply draw $\nu$ directly from $\rho$ and count the number of satisfying assignments, then our statistical estimate could have exponentially high variance and we would get an exponential-time algorithm. The efficient Monte Carlo sampling algorithms, like ours, are instead based on the following random process. Suppose we sample a clause $C_s \in \Phi$ 
with probability proportional 
to its weight $\rho(C_s)$. We then draw an assignment $\nu$ by setting all the literals in $C_s$ to be true, and drawing all variables 
not involved in $C_s$ from their original distribution $\rho$. Finally, we define $L$ to be the total 
number of satisfied clauses under $\nu$. Note that $L \geq 1$ since clause $C_s$ is satisfied. For this process we have
$$
\mu = \mathbb E[1/L] \cdot \rho(\Phi) = \mathbb E[1/L] \sum_{C \in \Phi} \rho(C).
$$

Here $\rho(\Phi)$ can be computed exactly. Thus, in order to estimate $\mu$, the critical problem becomes the estimation of the quantity
$$
p = \mathbb E[ 1/L ].
$$

The main difference between the Monte Carlo algorithms is how to estimate $p$. One simple algorithm discussed in \cite{KLM} is to sample all variables and count $L$ directly; this already gives a FPRAS. To improve on this, KLM randomly selects clauses until finding a true clause; the resulting waiting time is a geometric random variable which can be used to estimate $p$. 

The stopping condition in KLM is determined by a 
``self-adjusting" mechanism: the number of trials $N$ is not fixed in advance, but is determined dynamically in terms of 
the waiting times of the main loop.

Our first contribution is to develop Algorithm 1, a novel \emph{lazy sampling} version of KLM (\newKLM), which minimizes the amount of 
sampling by delaying variable assignment until needed. 

\setcounter{algorithm}{0}
\begingroup
\begin{algorithm}
    \caption{L-KLM}
    \begin{algorithmic}
        \REQUIRE DNF formula $\Phi$ and PAC parameters $\eps, \delta$.
    \end{algorithmic}
    \begin{algorithmic}[1]
        \STATE $T \gets \frac{8(1+\eps)m\log(3/\delta)}{(1-\eps^2/8)\eps^2}$
        \STATE $Y \gets 0$ ; $N \gets 0$
        \WHILE{$Y<T$}
            \STATE Select clause $C_s$ with probability proportional to $\rho(C_s)$.
            \STATE Create a partial assignment $\nu$, which assigns only the variables in $C_s$ so that $\nu$ satisfies $C_s$. All other variables are left unassigned.
            %  (indeterminate).
	    \REPEAT
	    \STATE Select a clause $C_k$ uniformly with replacement.
            \STATE $Y\gets Y + 1$ 
            \STATE Update $\nu$ by randomly sampling all unassigned variables in $C_k$. 
\UNTIL{$\nu$ satisfies $C_k$} 
            \STATE $N \gets N + 1$
        \ENDWHILE
        \RETURN{estimates $\hat p = \frac{Y}{Nm}$ and $\hat \mu = \rho(\Phi) \hat p$.}
    \end{algorithmic}
\end{algorithm}
\endgroup

In order to sample clauses proportional to their weight efficiently, we use the data 
structure and algorithm
from \citep{sample}. The algorithm achieves $O(1)$ sampling time with $O(m)$ preprocessing time.

For the purposes of theoretical analysis, the entropy-optimal 
sampling algorithm of \cite{randomness} or \cite{randomness2} can be used to sample the variables from their probabilities $\rho(v)$. For consistency with prior benchmarks,  all our experiments are based on unweighted counting and so we did not implement this advanced algorithm. We used a simpler method where 
we generate a buffer of pseudorandom bits from Mersenne Twister, and use it bit-by-bit for each variable as 
needed.

\section{Main Algorithm}

Now we introduce our Main Algorithm, which features an adaptive stopping rule and a short-circuiting mechanism 
for clause evaluations. 
We refer to each iteration of the main loop (Lines 5 -- 14) as a \emph{trial}; each iteration of the inner loop (Lines 10 -- 12) 
is a \emph{step}. 

\begingroup
\begin{algorithm}
    \caption{Main Algorithm}
    \begin{algorithmic}
        \REQUIRE DNF formula $\Phi$ and PAC parameters $\eps, \delta$.
    \end{algorithmic}
    \begin{algorithmic}[1]
        \STATE Set $T$ to the minimum integer such that 
            $\bigl(\frac{e^{\frac{\eps}{1+\eps}}}{1+\eps} \bigr)^T + \bigl(\frac{e^{\frac{-\eps}{1-\eps}}}{1-\eps} \bigr)^T \leq \delta$.
        \STATE Choose a permutation $\pi$ over the clauses in $\Phi$ (see procedure P1).
	\STATE Prepare data structures in accordance with permutation $\pi$ (see Section~\ref{data-struct-sec})
        \STATE $Y \gets 0$ ; $N \gets 0$
        \WHILE{$Y < T$}
            \STATE Select clause $C_s$ with probability proportional to $\rho(C_s)$.
             \STATE Create a partial assignment $\nu$, which assigns only the variables in $C_s$ so that $\nu$ satisfies $C_s$. All other variables are left unassigned.
            %   (indeterminate).
	       \STATE Draw a random variable $R \in \mathbb Z_{\geq 1}$ with probability distribution $P(R = r) = \frac{1}{r(r+1)}$.
            \STATE $\ell \gets 1$
            \FOR{$C \in\Phi\setminus \{C_s\}$ in the order of $\pi$ \text{while} $\ell \leq R$}
                \STATE Update $\nu$ by randomly sampling all unassigned variables $v \in C$.
                \STATE \textbf{if} $\nu$ satisfies $C$ \textbf{then} $\ell \gets \ell + 1$
            \ENDFOR
            \STATE \textbf{if} $\ell \leq R$ \textbf{then} $Y \gets Y+1$
            \STATE $N \gets N + 1$
        \ENDWHILE
        \RETURN{\label{last-line}estimates $\hat p = \frac{T}{N}$ and $\hat \mu = \rho(\Phi) \hat p$.}
    \end{algorithmic}
\end{algorithm}
\endgroup

The algorithm begins by generating a permutation $\pi$ to order the clause evaluations. 
At each trial, it iterates through $\pi$ to count clauses which are true in an assignment $\nu$ (generated as needed), 
and has a probability of aborting early based on the count of satisfied clauses encountered. A successful trial is one that
was not aborted early and the algorithm terminates when the 
number of successes reaches a predetermined value $T = \Theta( \frac{\log(1/\delta)}{\eps^2} )$. Thus, both our stopping rule
and short-circuiting adapt to the value of $p$. In fact, the probability of a trial  resulting in a success
is precisely $p$, regardless of the permutation $\pi$ (see Lemma~\ref{gga}); when $p$ is small the average trial time is low and the amount
of total trials is high, but when $p$ is large the average trial time is high and the number of total trials is low.

Critically, the permutation $\pi$ is re-used over all trials. Procedure P1 generates $\pi$ by ``blending" a given heuristic clause permutation $\pi_{\text{h}}$ with a random permutation, 
based on a blending rate constant $\beta \in [0,1]$.  From this, we can generate a variety of data structures 
which are specialized for that order and develop heuristics for $\pi$ to process true clauses earlier on average compared to a random order. 
This is one of the key differences between our Main Algorithm and KLM, which requires a randomly 
selected clause at each step. 

\begingroup
\begin{procedure}
    \caption{}
    \begin{algorithmic}
\REQUIRE Set of clauses $\Phi$, and permutation $\pi_{\text{h}}$ of $\Phi$.
    \end{algorithmic}
    \begin{algorithmic}[1]
  	\STATE Initialize $\mathcal C \leftarrow \Phi$.
        \FOR{$j=1$ \TO $m$}
	    \STATE Let $k$ be the smallest index with $\pi_{\text{h}}(k) \in \mathcal C$, and let $v' = W(\pi_{\text{h}}(k))$.
	    \STATE Let $w' = W(\mathcal C)/|\mathcal C|$ be the average width of the clauses in $\mathcal C.$
	    \STATE With probability $\beta \min\{1, w'/v' \}$, set $\pi(j) = \pi_{\text{h}}(k)$,  \\
 \ \ \ \ else set $\pi(j)$ to be a random clause in $\mathcal C$;  
 	    \STATE Update $\mathcal C \gets \mathcal C \setminus \{ \pi(j) \}$.
        \ENDFOR

    \end{algorithmic}
\end{procedure}

Any heuristic clause permutation $\pi_{\text{h}}$ can be chosen, however in the rest of the work we choose $\pi_{\text{h}}$ to be
ordered by increasing clause width.
The motivation is to find true clauses as quickly as possible; via the short-circuit mechanism, this will reduce the cost of each trial. 
Furthermore, clauses with small width are faster to check and more likely to be true.
On the other hand, we should partially randomize the clause selection, to avoid getting trapped by worst-case instances with many 
small false clauses.
The procedure P1 with our heuristic for $\pi_{\text{h}}$ attempts to balance these goals. For experiments we chose $\beta = 0.99$; see Section~\ref{empirical-sec} for further discussion.

\section{Analysis of the Algorithm}
We begin the analysis by showing correctness. Let us say that a trial \emph{succeeds} if $Y$ is incremented, i.e. $\ell \leq R$ and 
the loop does not abort. We define random variable $N$ to be the value at the 
end of the main algorithm (Line \ref{last-line}).

\begin{observation}
\label{ggb}
For $\eps, \delta \in (0,1)$, there holds

$$ 
\frac{\log(2/\delta)}{\log(1-\eps) + \frac{\eps}{1-\eps}} \leq T \leq \frac{\log(2/\delta)}{\log(1+\eps) - \frac{\eps}{1+\eps}}.
$$

In particular, for $\eps, \delta \in (0,3/4)$, we have $T = \Theta( \frac{\log(1/\delta)}{\eps^2} )$.
\end{observation}

The following key lemma explains the statistical basis for the algorithm: 
\begin{lemma}
\label{gga}
The probability that a given trial succeeds is $p$, irrespective of the permutation $\pi$.
\end{lemma}
\begin{proof}
Suppose that, in a given  trial, we condition on the chosen clause $C_s$ and the variable assignment $\nu$, with $L \geq 1$ satisfied clauses. The running counter $\ell$ will eventually cover all the true clauses, unless it aborts early. So the trial succeeds if and only if $R \geq L$. This has probability $\sum_{r=L}^{\infty} \frac{1}{r(r+1)} = 1/L$.

Taking the expectation over $\nu$, the success probability is $\E[1/L] = p$.
\end{proof}

As a result, we can use standard results and concentration bounds to analyze the algorithm performance.  
\begin{theorem}
\label{result1}
 For any $\eps,\delta \in (0,3/4)$, Algorithm 2 provides an $(\eps,\delta)$ approximation for $\mu$, and the expected value of $N$ is $\Theta( \frac{\log(1/\delta)}{p \eps^2} )$.
Furthermore, these bounds all hold even conditioned on an arbitrary choice of permutation $\pi$.
\end{theorem}

\begin{proof}
By Lemma~\ref{gga}, the final value $N$ is a negative-binomial random variable with $T$ successes and probability $p$. By Observation~\ref{ggb}, we have $T = \Theta( \frac{\log(1/\delta)}{\eps^2} )$. So $\E[N] = T/p = \Theta( \frac{\log(1/\delta)}{p \eps^2} )$.

To show that $\hat p$ satisfies the PAC bounds, let us first show $$
\Pr\big( \hat p < p (1 - \eps) \big)< \Big(\frac{e^{\frac{-\eps}{1-\eps}}}{1-\eps}\Big)^T.
$$

Instead of running the algorithm for a fixed number of successes $T$, 
suppose we ran it for a fixed number of trials $\tilde N = \big \lfloor \frac{T}{p (1-\eps)} \big \rfloor$; 
let $\tilde Y$ denote the resulting number of successes at that point. So $\tilde N$ is a scalar quantity, 
and $\tilde Y$ is a binomial random variable with success rate $p$. Observe that we have $\hat p < p(1-\eps)$ 
if and only if $N > \tilde N$, which in turn holds if and only if $\tilde Y \leq T-1$. This has probability at most 
$F^{-} ( \tilde N p, T - 1)$,  
where $F^{-}(x,y) = e^{y-x} (y/x)^y$ is the Chernoff lower-tail 
bound that a sum of 
independent variables with mean $x$ is at most $y$. 

Via monotonicity properties of $F^{-}$ and the fact that $\tilde N \geq \frac{T}{p(1-\eps)} - 1$, this is at most $$
F^{-} \Bigl( \frac{T}{1-\eps}, T \Bigr) = \Big(\frac{e^{\frac{-\eps}{1-\eps}}}{1-\eps}\Big)^T.
$$

A symmetric argument with Chernoff upper-tail bounds gives 
\[\Pr\big( \hat p > p (1 + \eps) \big)< \Big(\frac{e^{\frac{\eps}{1+\eps}}}{1+\eps}\Big)^T.\]

By a union bound, the overall probability 
that $p(1-\eps) \leq \hat p \leq p (1+\eps)$ is at least $1 - \bigl(\frac{e^{\frac{-\eps}{1-\eps}}}{1-\eps}\bigr)^T - \bigl(\frac{e^{\frac{\eps}{1+\eps}}}{1+\eps}\bigr)^T$, 
which is at least $1 - \delta$ by specification of $T$.
\end{proof}

We next show the complexity bounds for the algorithm. 
\begin{theorem}
\label{theorem2}
Let all values $\rho(v)$ be in the range $[b, 1-b]$ for $b \in (0,1/2]$.  Then each trial of the Main Algorithm has expected work $O\big( m w p \log(2/p) \big)$ and expected randomness complexity $O\big(m p  \log(2/p) \log(1/b)\big)$.
\end{theorem}
\begin{proof}
See Appendix.
\end{proof}

\begin{theorem}
\label{main-analysis-theorem}

Let $\eps,\delta\in(0,3/4)$, and all the values $\rho(v)$ be in the range $[b, 1-b]$.
Then the expected work of the Main Algorithm is $$
O \Bigl( \frac{m w \log(2/p) \log(1/\delta)}{\eps^2} \Bigr),
$$
and the expected randomness complexity is $$
O \Bigl( \frac{\min\{ n/p, m \log(2/p) \log(1/b) \} \log(1/\delta)}{\eps^2} + m \log m \Bigr).
$$
\end{theorem}\endgroup

\begin{proof}
We first show the bound on work. Let $S$ denote the work of a given trial; note that, when  $\pi$ is fixed, all trials are independent and identically distributed. By Wald's equation (\cite{wald})
the overall expected work conditional on $\pi$ is $\E[N \given \pi] \cdot \E[S \given\pi]$. 
By Theorem~\ref{result1} we have $\E[N \given \pi] = O( \frac{\log(1/\delta)}{p \eps^2} )$. So the overall expected work, conditional on $\pi$, 
is $O \bigl( \frac{\E[S \given \pi] \log(1/\delta)}{p \eps^2} \bigr)$. Now take the expectation over $\pi$ and note that Theorem~\ref{theorem2} gives $\E[S] \leq O( m w p \log(2/p)  )$.

The bound on randomness complexity is completely analogous, except that we also require $O(m \log m)$ random bits to generate $\pi$ in algorithm P1, and 
we note that the maximum amount of randomness in each trial is $O(n + \log m)$ to choose $C_s$ and sample each variable. We can assume that $m \leq 3^n$ and hence the $\log m$ term is negligible, since each clause can be determined by choosing whether a variable appears negatively, positively, or not at all. 
\end{proof}

To interpret this bound, note that for realistic problem instances, $b$ is typically chosen to be constant; for example, in unweighted counting we have $b = 1/2$. The parameter $p$ is harder to interpret; we always have $p \in [1/m,1]$, and for some problem settings, $p$ can be relatively large, or even constant. For example, the reliability sampling algorithm of \cite{harrissrin} has $p = \Omega(1)$. For other problem settings, including many of the previous benchmark instances, we would typically have $p = \Theta(1/m)$. 

For sake of completeness, we also state a formal result for  L-KLM. The analysis is very similar to our Main Algorithm so we omit the proofs.
\begin{theorem}
\label{main-analysis2}
Let $\eps, \delta \in (0,3/4)$ and all values $\rho(v)$ be in the range $[b, 1-b]$ for $b \in (0,1/2]$. 
Then L-KLM has expected work $O( \frac{m w \log(1/\delta) }{\eps^2})$ and expected randomness complexity $O( \frac{m \log(m/b) \log(1/\delta)}{\eps^2}).$
\end{theorem}

\section{Data structures and implementation notes}
\label{data-struct-sec}
Here, we describe some of the implementation details for how we deal with the clauses in each step. 
We suppose that $\pi$ has been generated and is fixed throughout, and we write $C_i = \pi(i)$.

As a starting point, we will simply store all the clauses consecutively in order $C_1, C_2, \dots,$ so that iterating through the clauses in 
each step becomes a stride-one operation. The data structure used for sampling $C_s$ can be generated \emph{after} this step, so that we never need to use the original clause ordering afterward.

The next preprocessing step is to sort the variable indices so that the 
clause variable sequence $(C_1, C_2, \dots,)$ is 
lexicographically minimal. Concretely, if clause $\pi(1)$ has width $k_1$, then we relabel its variables to 
indices $1, \dots, k_1$. If clause $C_2$ contains $k_2$ variables disjoint from $C_1$, we relabel them $k_1 + 1, \dots, k_1 + k_2$. 
We also sort the variables within each clause. These reordering steps help with memory locality for the variables.

The variable reordering has a second subtle but powerful effect regarding the storage of $\nu$. 
Our Main Algorithm, as in Pepin, maintains $\nu$ using two bit-arrays of size $n$, one for the value of each variable 
and a second to record which variables have already been assigned.  The second array must be reset in each trial. As discussed in \cite{pepin}, there are two methods that can be used for this: the first is to go through the assigned positions and zero out the bits one-by-one (the ``sparse" method.) The second is to use a single large-scale memory operation such as $\texttt{memset}$ to clear the array completely (the ``dense" method).

The sparse method respects the asymptotic bounds, but it is slow and requires additional data structures to track which variables have been assigned. The dense method is very fast, and can use built-in SIMD operations; the downside is that it is still an $O(n)$ operation for each trial, which essentially negates the benefit of lazy sampling. An optimized implementation of L-KLM should switch between the dense and sparse methods, depending on the size of $p$ and the capabilities of the underlying processor. In our experiments, we tried both implementations; we consistently found that the sparse method was faster for our benchmarks and we used it throughout for the L-KLM code.

Our Main Algorithm uses a simpler scheme with the benefit of both methods: whenever a trial is finished after some step $j \leq m$, we use $\texttt{memset}$ to 
zero out the variable assignment array in the contiguous block of positions $0, \dots, t$, where $t$ is the maximum variable index over clauses $C_1, \dots, C_{j}$. The latter can  be precomputed as part of our data structure. Critically, because of the variable reordering, the $\texttt{memset}$ has cost $O(t) \leq O(W(C_1) + \dots + W(C_j))$ (with a very small constant), matching the asymptotic analysis. We also zero the positions corresponding  to clause $C_s$ using a sparse bit-by-bit representation.

\section{Empirical evaluation and discussion}
We chose $\beta=0.99$ for our Main Algorithm and benchmarked on unweighted DNFs in all experiments.

\label{empirical-sec}

\subsection{Comparison with Neural\#DNF and previous benchmarks}
As a baseline, we compared our Main Algorithm directly to reported Neural\#DNF running time and previous benchmark regimes in Figure~\ref{NN}. We generated 
DNFs with $n=\,$15,000 (the largest reported by Neural \#DNF) and $m=0.75n$, 
with various uniform widths $w = 3 ,5, 8, 13, 21, 34$. Such DNFs 
are somewhat artificial, with the undesirable property that $\mu \rightarrow 1$ as $n \rightarrow \infty$, 
but this matches methodology used in the experiments of \cite{FPRAS} and \cite{pepin} along with the 
parameter modifications in \cite{GNN}. 

\vspace{-.4cm}

\begin{figure}[ht]
    \centering
    \includegraphics[width=.6\linewidth]{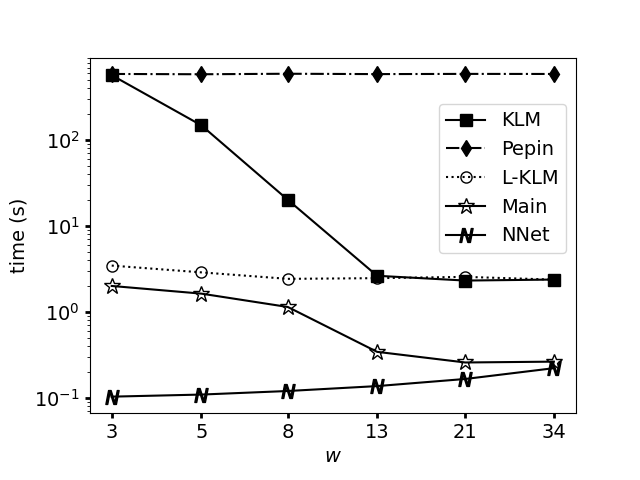}
    \caption{\label{NN} Work comparison of the competing methods. $n=\,$15,000, \\$m=\,$11,250, $\eps=0.1$, and $\delta=0.05$.}
\end{figure}

Neither Pepin nor L-KLM is significantly affected by the width $w$, as noted in \cite{pepin}. 
We speculate that the discrepancy in performance between KLM and Pepin from previous publications is
due to the tighter PAC bound requirement. Loose bounds $\eps=0.8$ and $\delta=0.36$ were used to benchmark Pepin; the $\eps,\delta$ scaling observed from Figure~\ref{fig:eps5000} as well as the extra $\log(1/\eps)$ factors in the 
complexity bounds of Table 2 indicate that Pepin will become much slower for the tighter 
bounds $\eps=0.1, \delta=0.05$.

We mention that most previous benchmarks had reported CPU time, but we found that this misrepresents the 
actual time required. The following experiments shown in Figure~\ref{fig:CPU} show the mean of both
wall-clock and CPU time for each method to solve the six benchmark regimes. Because of this issue, aside from Figure~\ref{fig:CPU} all other timing reported is in wall-clock time as opposed to CPU time.

\vspace{-.4cm}

\begin{figure}[ht]
    \centering
    \begin{subfigure}[b]{0.48\textwidth}
        \centering
        \includegraphics[width=.95\linewidth]{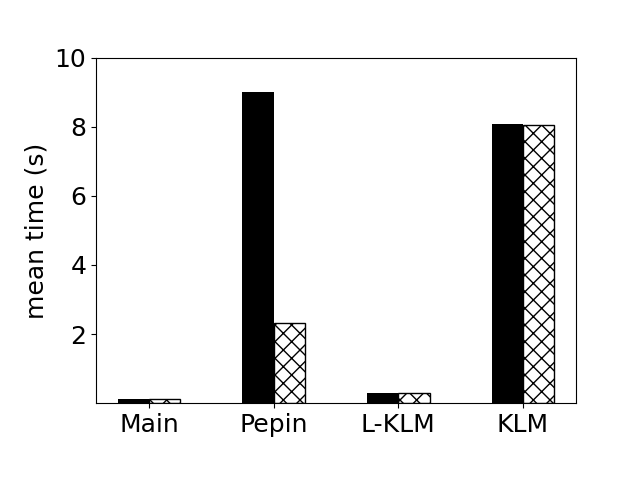}
        \caption{Mild PAC bounds.\\ $\eps = 0.2,$ $\delta = 0.1$,
        $n=m=2^{13}$.}
        \label{fig:cA}
    \end{subfigure}
    \begin{subfigure}[b]{0.48\textwidth}
        \centering
        \includegraphics[width=.95\linewidth]{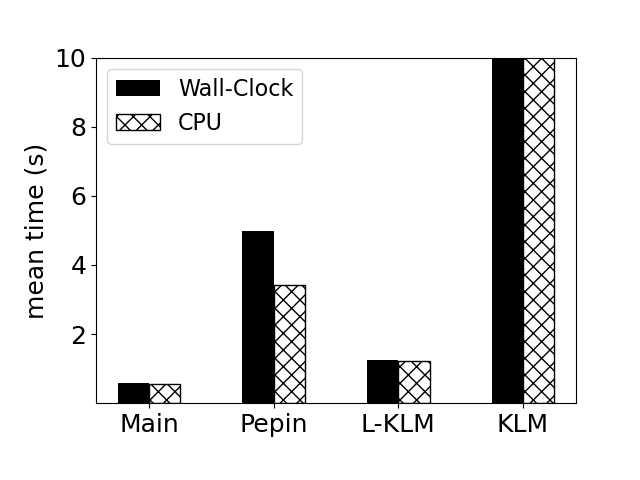}
        \caption{Loose PAC bounds.\\ $\eps = 0.5,$ $\delta = 0.25$,
        $n=m=2^{17}$.}
        \label{fig:cO}
    \end{subfigure}
    \caption{Mean running times for the two regimes on the six different scenarios of uniform width in \{3,5,8,13,21,34\}.}
    \label{fig:CPU}
\end{figure}

% \vspace{-.4cm}

We see that Pepin spends a large proportion of its time
waiting on the kernel and that this percentage increases as $\eps$ and $\delta$ decrease. 
We speculate this may be due to the larger
memory and randomness requirements inherent in Pepin. This provides good reason for us to record our benchmark in wall-clock
time.

\subsection{A new family of benchmarks}
One main disadvantage of prior benchmarks, where the clauses are generated independently, is that for large instances the satisfaction probability $\mu$ approaches to one. For such problems, much simpler algorithms based on naive Monte Carlo estimation could be used instead. 

To attempt to better capture hard large-scale problem instances, we generate a new family of synthetic formulas inspired by \cite{FPRAS} and \cite{GNN}. Each formula is randomly generated with a 
specified $n$ and $m$ beforehand, and has a small number of ``stems", i.e. a selection of randomly chosen variables
shared among many clauses. For each distinct clause, the amount of
variables added to its stem is randomly chosen. Finally, each variable 
in the clause is set according to $\rho$. Details are found in Procedure P2 below. Here, parameter $\alpha$ specifies the number of stems, $\gamma$ specifies the number of variables per stem, and $\lambda$ controls how many variables to add to a stem.

\begingroup
\begin{procedure}
    \caption{}
    \begin{algorithmic}
\REQUIRE Integers $n,m, \alpha,\gamma,\lambda>0.$
    \end{algorithmic}
    \begin{algorithmic}[1]
\STATE Initialize $\Phi \leftarrow \emptyset$.
    \WHILE{$|\Phi| < m$}
    \STATE Form a stem $C_{\sigma}$ by choosing $\gamma$ variables at random and assigning them according to $\rho$.
\FOR{$t = 1, \dots, m / \alpha$}
        \STATE Choose a width $w'$ uniformly at random in $\{1, \dots, \lambda \}$.
        \STATE $C\leftarrow C_\sigma$.
        \FOR{$i=1$ \TO $w'$}
            \STATE Randomly assign a randomly selected variable $v$ according to $\rho$.
            \STATE Add $v$ or $\neg v$ to $C$.
        \ENDFOR
        \STATE $\Phi\leftarrow\Phi\cup \{C \}$
        \ENDFOR
    \ENDWHILE
\RETURN DNF formula $\Phi$.
    \end{algorithmic}
\end{procedure}

There are two main advantages with the DNFs generated by this process. First is that non-uniform width
more accurately represents problems which show up in practice, and second we want to maintain a significant value for $p$ to ensure
a similar computational hardness as that of the small uniform width case.

All experiments were conducted on a compute node consisting of a $4\times$ 
Intel Xeon Platinum 8260 CPUs with $4\times24$ 
real cores and 1 TB of RAM.
All algorithms used a few GB, nowhere near 1 TB.
The KLM and \newKLM algorithms were implemented in C++ and compiled with the O3 flag.  The Monte Carlo algorithms, 
including ours, all share the same support functions wherever applicable
and use the improved discrete sampling algorithm with Mersenne Twister randomization. 
We obtained the code for Pepin
from the author's GitHub (version 1.3).
For complete transparency, we include Figure~\ref{overlap0} comparing the values of $p$ and $\mu$ from previous benchmarks.

% \vspace{-.4cm}

\begin{figure}[ht]
    \centering
    \begin{subfigure}[b]{0.48\textwidth}
        \centering
        \includegraphics[width=.95\linewidth]{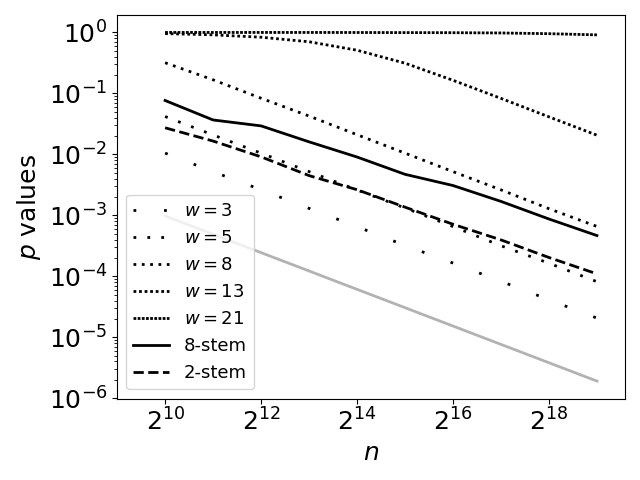}
        \caption{Comparison of $p$ value}
        \label{overlap1}
    \end{subfigure}
    \hfill
    \begin{subfigure}[b]{0.48\textwidth}
        \centering
        \includegraphics[width=.95\linewidth]{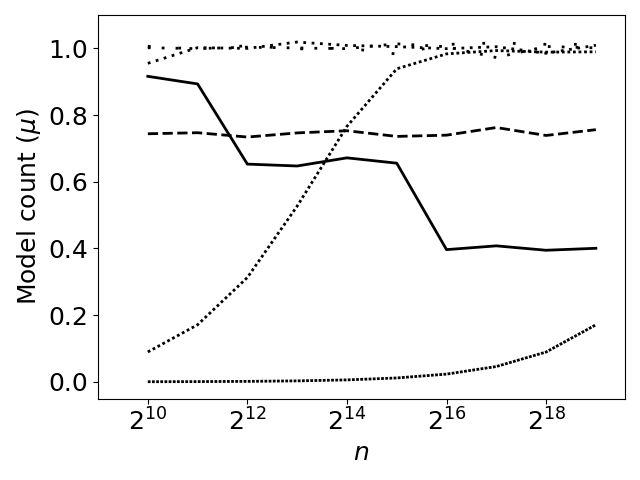}
        \caption{Comparison of model count}
        \label{count1}
    \end{subfigure}
    \caption{\label{overlap0}Comparison of benchmarks $p$ and $\mu$.}
\end{figure}

\subsection{Scaling DNF size}

We generated DNFs with Procedure P2 for $n=2^{x}: x = 10, \dots, 20$, and parameters $$
m = n, \alpha = 2, \gamma = \lfloor \tfrac{1}{10} \log_2 m \rfloor, \lambda = \lfloor 2\log_2 m \rfloor.
$$
The results are shown in Figure~\ref{size}.

\vspace{-.4cm}
\begin{figure}[H]
    \centering
    \includegraphics[width=.6\linewidth]{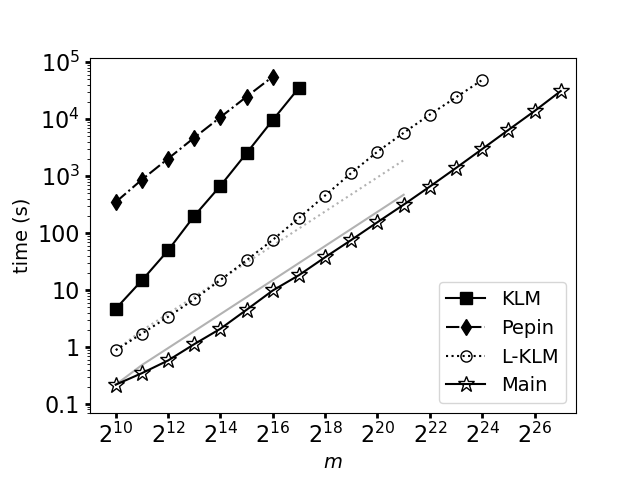}
    \caption{\label{size}Average work as problem size increases, $\eps = \delta = 0.05$.
        Solid and dotted gray lines indicate an ideal scaling.}
\end{figure}

We see that KLM has the worst scaling of all methods, and that Pepin solved the fewest problems before the timeout of $1$ day.
The experiment demonstrates that our Main Algorithm and \newKLM can solve problem sizes with $n > 10^6$ at $95$\% confidence. 
Indeed, our Main Algorithm's work at $n=2^{27}$ is comparable to KLM's work at $n=2^{17}$;
this indicates a significant improvement over Neural\#DNF's performance with $n=\,$15,000
at $90$\% confidence. Our Main Algorithm also provides more than an order of magnitude speedup compared to L-KLM.

\subsection{Scaling the PAC bounds}

For our next  experiment, we fixed $n = m = 2^{12}$, and fixed either $\eps$ or $\delta$ while the other varied over a wide range. 
See Figure~\ref{accSc}. Note that our Main Algorithm has much better scaling in $\eps, \delta$ compared to Pepin, and it
obtains $99.9\%$ confidence in comparable time to KLM and Pepin's attainment of $99.0\%$ and $95.0\%$ confidence respectively.

To test confidence we generated $32$ DNFs with number of variables $n$ ranging from $4$ to $32$ in increments of $4$. For each
$n$, we allowed $m$ to range through $\{\frac{3}{4},\frac{4}{4},\frac{5}{4},\frac{6}{4}\}\cdot n$. We fixed $\delta = 0.05$
and allowed $\eps$ to range in $[0.005,0.1]$.
Figure~\ref{acc} shows that all the methods provide the guaranteed confidence, except for Pepin in smaller problem cases.
There seems to be some dependency between the size of the problem and $\eps$ for the confidence guarantee in Pepin.
We observe that our Main Algorithm returns a relative confidence well below the guarantee regardless of problem size.

\vspace{-.4cm}
\begin{figure}[ht]
    \centering
    \begin{subfigure}[b]{0.48\textwidth}
        \centering
        \includegraphics[width=.95\linewidth]{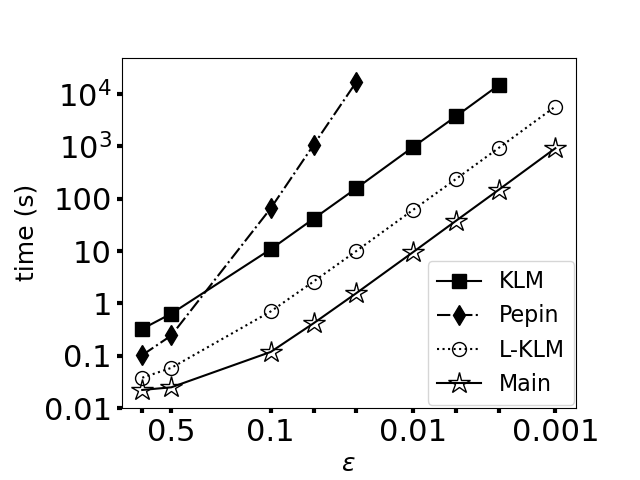}
        \caption{Average work as $\eps$ decreases from $0.8$ to $0.001$. $\delta=0.1$.}
        \label{fig:eps5000}
    \end{subfigure}
    \hfill
    \begin{subfigure}[b]{0.48\textwidth}
        \centering
        \includegraphics[width=.95\linewidth]{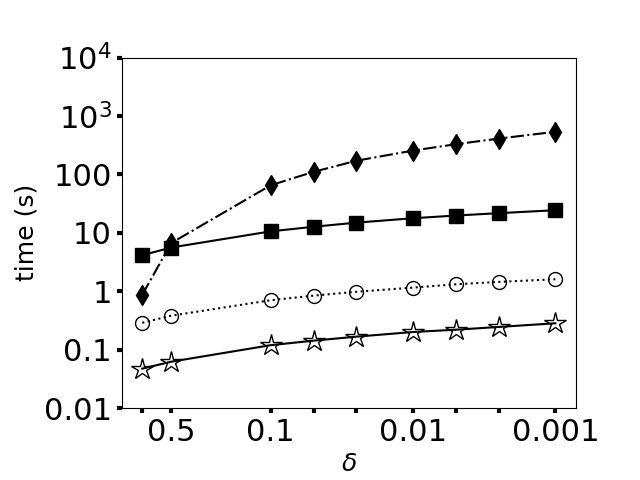}
        \caption{Average work as $\delta$ decreases from $0.8$ to $0.001$. $\eps=0.1$.}
        \label{fig:delta5000}
    \end{subfigure}
    \caption{\label{accSc}$(\eps,\delta)$-Scalability. $n=m=2^{12}$}
\end{figure}

\vspace{-.3cm}
\begin{figure}[ht]
    \centering
    \includegraphics[height=.45\linewidth, width=.75\linewidth]{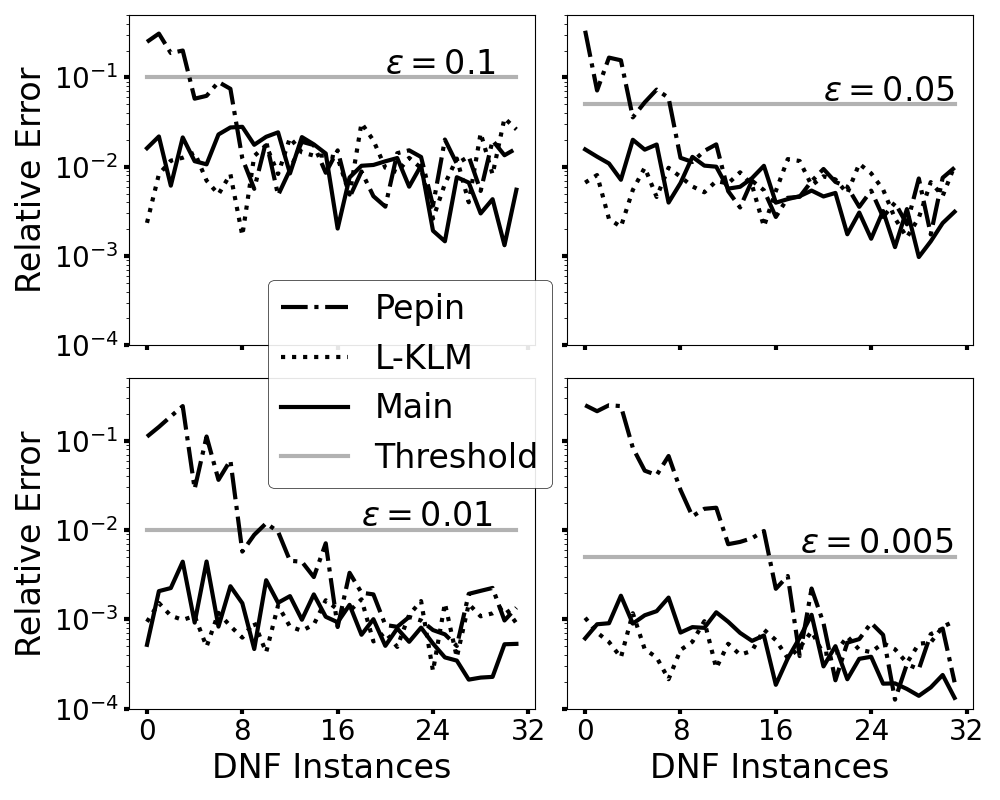}
    \caption{\label{acc}Confidence Tests. $\delta = 0.05, \eps \in [0.005,0.1]$}
\end{figure}

\subsection{Discussion on algorithm performance}

The fact that
we outperform all other FPRASs even on uniform-width DNFs and on small problem sizes, 
where our permutation heuristic provides little to no benefit, is
a testament to the robustness of our approach.

There appear to be a number of distinct reasons for the improved perfomance, with different effects on different types of datasets and different parameter regimes. Here we highlight some factors that we believe explain many of the advantages of our Main Algorithm.

\vspace{-2pt}
\paragraph{Random sampling.} The inner loops of KLM and L-KLM sample $\Theta(\log m)$ random bits per step, to choose the next clause. 
By contrast, in our Main Algorithm, the only random samples are needed for the variables, aside from drawing $R$ and $C_s$ (which are negligible). 
Because we can use short-circuit evaluation of each clause, this only requires $O(1)$ random bits in expectation.
Consequently, our Main Algorithm uses much less randomness compared to KLM and L-KLM. 
While theoretical analysis of algorithms typically assumes access to a random tape in $O(1)$ time, 
in practice generating (pseudo)-random values can be much more expensive than other arithmetic operations. 

In Figure~\ref{randomfig}, we generated DNFs with procedure P2 and chose $\alpha = 8$, $\gamma = \lfloor \tfrac{1}{4}\log_2 m
 \rfloor$, and $\lambda = \lfloor 2\log_2 m \rfloor$. We see the significantly reduced and slowly-growing random number generation in the algorithm.

\vspace{-.4cm}
\begin{figure}[H]
    \centering
    \includegraphics[width=.6\linewidth]{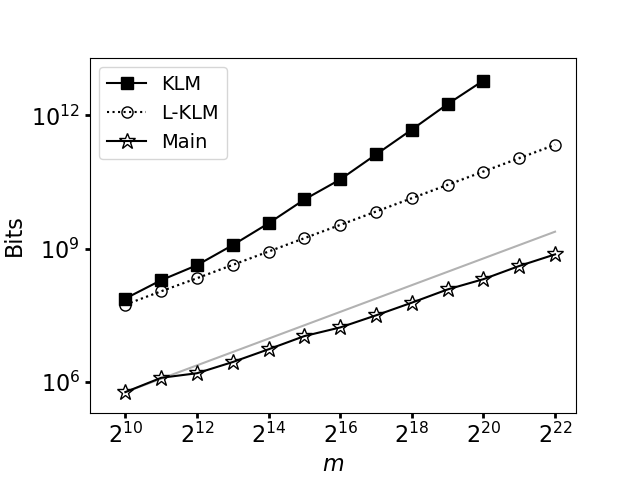}
    \caption{\label{randomfig}Number of random bits needed as problem size increases. $\eps = 0.1$, $\delta = 0.5$, and $n=m$.
        Solid gray line indicates ideal scaling.}
\end{figure}

\vspace{-20pt}
\paragraph{Memory access pattern.} The inner loop of our Main Algorithm is to step through the clauses in a fixed order given by the permutation $\pi$. By contrast, in KLM and L-KLM, the clauses must be processed in a \emph{uniformly random} order. This is one of the worst access patterns for large memory. While the usual asymptotic conventions assume that memory access takes $O(1)$ time, real-world memory times typically vary depending on access locality.

The situation for Pepin is even worse: while KLM, L-KLM and our Main Algorithm handle each trial separately, Pepin must effectively store them all together. As a result, the memory usage scales with problem size as well as $\eps$ and $\delta$. For small values of $\eps$, in particular, the working memory may exceed the available low-level memory caches.  This can dramatically inflate the running time, and is another reason why our Main Algorithm can be so much faster compared to Pepin.

\vspace{-2pt}
\paragraph{Clause selection heuristic.} One key feature of our Main Algorithm is that the permutation $\pi$ can use heuristics to choose the clause, e.g. by preferentially selecting small clauses. In our asymptotic analysis, we have ignored any potential benefits from the clause-selection heuristic; it is simply taken as a possible surprise bonus in work. As we see in Figure~\ref{beta}, the clause selection with $\beta = 0.99$ provides almost all of the theoretical benefits of the random permutation on pathological instances, but also essentially matches the greedy clause-selection heuristic on typical instances.

% \vspace{-.4cm}
\begin{figure}[ht]
    \centering
    \includegraphics[height=.35\linewidth, width=.9\linewidth]{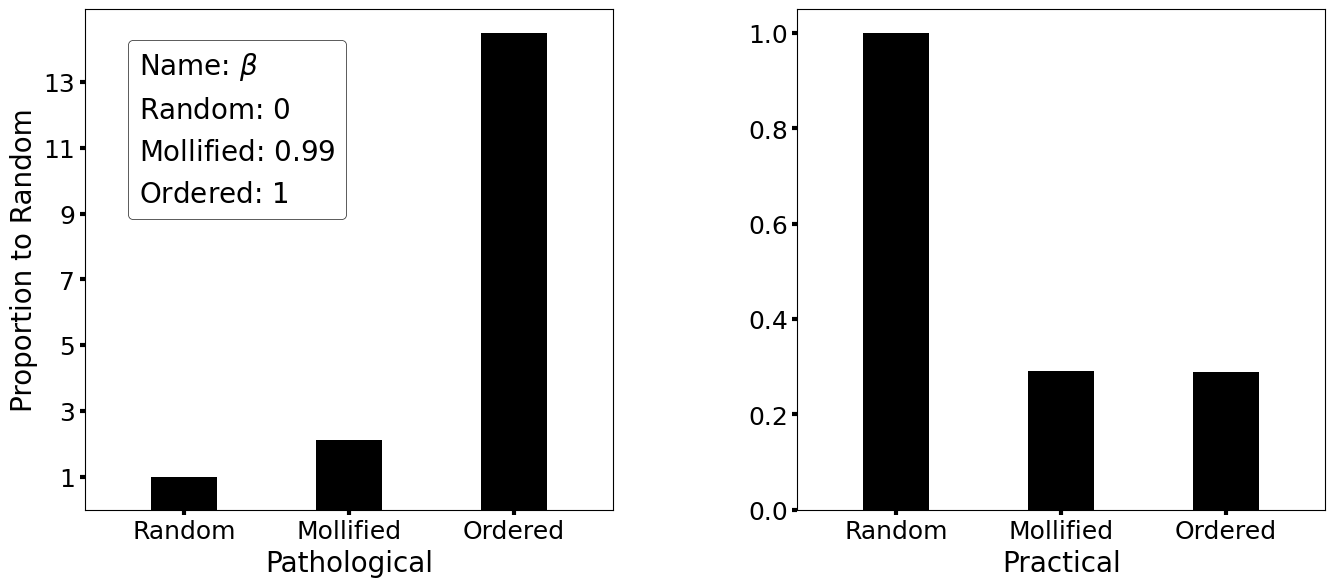}
    \caption{\label{beta}Parameter selection of $\beta$. Pathological case has half of the clauses a given width
        and the rest are width+1; furthermore, clauses of different widths are never simultaneously true. Practical case is the same
        as other benchmarks.}
\end{figure}

% \vspace{-2pt}

\vspace{-14pt}
\paragraph{Variable ordering.}  A side effect of the consistent clause order is that we can use more compact data structures to represent the variables. This allows further speedups. It also allows us to simplify the code, since we do not need to switch between sparse and dense representations to clear the variables.

\vspace{-2pt}
\paragraph{Tight PAC Bounds.} In both KLM and our Main Algorithm, the average number of trials is $\Theta( \frac{\log(1/\delta)}{p \eps^2} )$. However, the constant factor for KLM is roughly $4 \times$ larger. The main reason for this is that our stopping threshold is defined in terms of Binomial random variables and a Chernoff bound; this is similar to, but slightly tighter than, the generic stopping rule of \cite{optimalMC}. By contrast, the threshold for KLM requires more generic bounds in terms of optional stopping for martingales; note that the stopping rule of \cite{optimalMC} would not be valid for the KLM algorithm, since the relevant  statistic for the KLM algorithm (the waiting time to observe a true clause) is not bounded.

We also mention that the concentration bounds given for Pepin may be erroneous, as shown in Figure 3.

\appendix

\section{Proof of Theorem~\ref{theorem2}}
Since the trials are i.i.d. given $\pi$, it suffices to consider just the first trial. 
By the principle of deferred decisions, suppose that we are generating $\pi$ step-by-step in an online fashion. 
We choose each clause $\pi(i)$ as we encounter it either, uniformly at random or deterministically from $\pi_{\text{h}}$. Let us refer 
to these as ``random" or ``heuristic" steps respectively, and we denote by $\mathcal C_r$ and $\mathcal C_h$ the set of clauses checked in random and heuristic steps. (For analysis purposes, we count $C_s$ in these sets; this overestimates the complexity because $C_s$ is skipped in the actual implementation). The loop terminates as soon as the total number of true clauses aside from $C_s$ exceeds $R-1$. Thus, $\mathcal C_r \cup \mathcal C_h$ may be much smaller than the full set of clauses $\Phi$.

\begin{proposition} 
\label{aprop1}
There holds
$$
\E[ |\mathcal C_h| ] \leq O(\E[ |\mathcal C_r| ]), \qquad \text{and} \qquad \E[W(\mathcal C_h) ] \leq O(\E[ W(\mathcal C_r) ]).
$$
\end{proposition}
\begin{proof}
Suppose we are at some step $i$ in the generation process, and we have not yet terminated the loop; 
let $v' = W(\pi_{\text{h}}(k))$ and let $w'$ be the average width of the remaining clauses. 
The probability of choosing the clause from $\pi_{\text{h}}$ is $q = \beta \min\{1, w'/v'\}$. Thus, the ratio of the subsequent clause's 
expected contribution of $W(\mathcal C_h)$ to $W(\mathcal C_r)$ is given by
$$
\frac{ q   v'}{ (1-q) w' } \leq \frac{\beta (w'/v') \cdot v' }{ ( 1 - \beta) w' } = \frac{\beta}{1-\beta}.
$$

In a completely analogous way, the ratio of the subsequent clause's expected contribution of $|\mathcal C_h|$ to $|\mathcal C_r|$ is given by $\frac{q }{ 1-q } \leq \frac{\beta}{1-\beta}.$
\end{proof}

\begin{lemma}
\label{comb-lemma}
Let $q_1, \dots, q_k \in [0,1-b]$ for $b \in (0,1/2]$. Let $H(x)$ denote the binary entropy function
$$
H(x) = -x \log_2(x) - (1-x) \log_2(1-x).
$$

Then 
$$
\sum_{i=1}^k H(q_i) \prod_{j=1}^{i-1} q_j \leq O(\log(1/b)).
$$
\end{lemma}
\begin{proof}
It suffices to show this under the assumption that $b < 0.01$.  Let us define $S(q_1, \dots, q_k) = \sum_{i=1}^k H(q_i) \prod_{j=1}^{i-1} q_j$; observe that
$$
S(q_1, \dots, q_k) = H(q_1) + q_1 S(q_2, \dots, q_k).
$$

We show by induction on $k$ that $S(q_1, \dots, q_k) \leq 2 \log(1/b)$. The base case of the induction $k = 0$ is vacuous. For the induction step, we have
$$
S(q_1, \dots, q_k) \leq H(q_1) + q_1 \cdot 2 \log(1/b).
$$

Consider the function $x \mapsto H(x) + 2 x \log(1/b)$.  It has a positive derivative as $x 
\rightarrow 0$, and has a unique critical point at $x = \frac{1}{1+2^{2 \log b}} > 1-b$. Hence it is increasing on the range $[0,1-b]$, and in particular we have $$
S(q_1, \dots, q_k) \leq H(q_1) + 2 q_1 \log(1/b) \leq H(1-b) + 2(1-b) \log(1/b).
$$

Routine calculations show that $H(1-x) + 2 (1-x) \log(1/x) \leq 2 \log(1/x)$ for $x < 0.01$.
\end{proof}

\begin{lemma}
\label{aprop3}
For a given clause $C = \pi(i)$, suppose that $\rho(a) < 1 - b$ for $b \in (0,1/2)$. Then the lazy sampling for $C$ can be implemented to use $O( \log(1/b) )$ random bits and $O(W(C))$ time complexity in expectation, conditioned on all prior random values including $\pi(1), \dots, \pi(i-1)$.
\end{lemma}
\begin{proof}
Let $a_1, \dots, a_k$ denote the literals in $C$ which have not yet been assigned in $\nu$; correspondingly define $q_i = \rho(a_i) \in [0,1-b]$ for $i = 1, \dots, k$. For the lazy sampling, we will successively sample a value $X_i \sim \text{Bernoulli}(q_i)$ and include either the literal $a_i$ or the literal $\neg a_i$ depending on whether $X_i = 1$ or $X_i = 0$ respectively. We stop the procedure when we encounter a literal which is false under $\nu$. Let $K$ denote the number of sampling steps; so $K$ is either the first index with $X_K = 0$, or $K = k$ if $C$ is true under $\nu$.  Since each literal has a probability of at least $b$ of being false, we have $\E[K] \leq \min\{k, 1/b \}$. 

We draw the variables $X_i$ using the algorithm of \cite{randomness2}. In expectation, this uses  $\E[\sum_{i=1}^{K} H(X_i)] + O(\log K )$ random bits and $O( \E[K]  )$ work. (As usual, we assume that basic arithmetic operations can be performed in constant time). We calculate here
$$
\E \Bigl[ \sum_{i=1}^{K} H(X_i) \Bigr] =  \sum_{i=1}^{k} H(q_i) \prod_{j=1}^{i-1} q_j.
$$

As shown in Lemma~\ref{comb-lemma}, this is at most $O(\log (1/b) )$.
\end{proof}

\begin{proposition} 
\label{aprop2}
If we condition on $C_s$ and the full variable assignment $\nu$, then each clause $C$ has
$$
\Pr( C \in \mathcal C_r ) \leq \frac{2(1 + \log L)}{L}. 
$$
\end{proposition}
\begin{proof}
 We want to estimate the probability that clause $C$ is chosen for $\mathcal C_r$ before the loop aborts early.   We claim that, conditioned on $R$, there holds
\begin{equation}
\label{ql-eqn}
\Pr( C \in \mathcal C_r  \mid R) \leq \min\{1, \frac{2 R}{L} \};
\end{equation}
the claimed result holds after integrating over $R$. 

When $L \leq 2$ then Eq.~(\ref{ql-eqn}) holds vacuously since $R \geq 1$ and the RHS becomes equal to one.  So suppose $L \geq 3$. Let $C'_1, \dots, C'_{L-2}$ be any other true clauses aside from $C$ and $C_s$. Consider the following coupling experiment: we first generate a uniformly random permutation $\tilde \pi$; whenever we execute a random sampling step, we choose the next available element from $\tilde \pi$. Let $\tilde M$ denote the number of clauses $C'_j$ that appear before $C$ in $\tilde \pi$, and let $M$ denote the total number of true clauses, from both $\mathcal C_r$ and $\mathcal C_h$, sampled before $C$.  Note that if $C$ is 
chosen for a random sampling step, then $M \geq \tilde M$. So $C \in \mathcal C_r$ holds only if $\tilde M < R - 1$.

Since $\tilde \pi$ is a uniform permutation, $\tilde M$ is a discrete uniform random variable in 
$\{0, \dots, L - 2 \}$. Hence $$
\Pr(\tilde M < R - 1) = \min\{1, \frac{R-1}{L-1} \} \leq \frac{2 R}{L}.
$$
 verifying Eq.~(\ref{ql-eqn}).
\end{proof}

We can show the claimed result of Theorem~\ref{theorem2}.

\begin{proof}[Proof of Theorem~\ref{theorem2}.]
For the work factor, there are two costs to the algorithm: dealing with the variables in $C_s$ (i.e. setting the variable values), and dealing with the clauses in each step. We examine them in turn.

First, setting the variables in $C_s$ has cost $O(W(C_s))$. Since each clause $C$ is chosen with probability proportional to $\rho(C)$, the overall expected cost here is
$$
\frac{\sum_C  \rho(C) W(C)}{\rho(\Phi)}.
$$

Note that $\mu \geq \rho(C)$ for any clause $C$. Hence this sum is at most
$$
\frac{\sum_C  \mu W(C)}{\rho(\Phi)} = \frac{m w \mu}{ \rho(\Phi) } = m w p.
$$

Next, by Lemma~\ref{aprop3}, the expected cost of the trial's sampling steps is at most
$$
O( \E[ W(\mathcal C_r \cup \mathcal C_h) ] ).
$$

Summing over all clauses $C$ and using Propositions~\ref{aprop1} and \ref{aprop2}, we have
$$
\E[ W(\mathcal C_r \cup \mathcal C_h) ] \leq O( \E[W(\mathcal C_r)] )
$$ 
and 
$$
\E[ W(\mathcal C_r) ] \leq \sum_{C} W(C) \E \Bigl[ \frac{2(1 + \log L)}{L} \Bigr] = 2 m w \E \Bigl[ \frac{1 + \log L}{L}  \Bigr].
$$

So overall the expected cost of the sampling steps is 
$$
O( m w ) \cdot \E \Bigl[ \frac{1 + \log L}{L}  \Bigr].
$$

Consider random variable $U = 1/L$. By definition we have $\E[ U ] = p$. 
Furthermore, since the function $f(x)  = x (1 + \log(1/x))$ is concave-down, Jensen's inequality gives
$$
\E \Bigl[ \frac{1 + \log L}{L} \Bigr] =  \E[ f(U) ] \leq f( \E[U ] ) = f(p) = p (1 + \log(1/p)),
$$
completing the desired time complexity bound.

The bound on randomness complexity is completely analogous. Note that, with the efficient sampling algorithm of \cite{randomness2}, we can sample $R$ and $C_s$ using $O(1)$ and $O(\log m)$ random bits in expectation.
\end{proof}

\bibliographystyle{named}
\bibliography{reference.bib}

\end{document}